\documentclass[a4paper,UKenglish,cleveref, autoref]{lipics-v2019}

\usepackage{microtype}
\usepackage{tcolorbox}
\usepackage{amssymb}
\usepackage{amsbsy}
\usepackage{amsmath}
\usepackage{bbm}
\usepackage{latexsym}
\usepackage{xspace}

\usepackage{wasysym}
\usepackage{graphicx}

\title{The Semialgebraic Orbit Problem}

\author{Shaull Almagor}{Department of Computer Science, Oxford University, UK}{shaull.almagor@cs.ox.ac.uk}{}{}
\author{Jo\"el Ouaknine}{Max Planck Institute for Software Systems,  
	Germany \and \\Department of Computer Science, Oxford
        University, UK}{joel@mpi-sws.org}{}{Supported by ERC grant
        AVS-ISS (648701), and by the Deutsche Forschungsgemeinschaft
        (DFG, German Research Foundation) --- Projektnummer 389792660
        --- TRR 248}
\author{James Worrell}{Department of Computer Science, Oxford University, UK}{jbw@cs.ox.ac.uk}{}{Supported by EPSRC Fellowship EP/N008197/1}

\authorrunning{S.\,Almagor, J.\,Ouaknine and J.\,Worrell}

\Copyright{Shaull Almagor, Jo\"el Ouaknine, and James Worrell}



\category{}

\relatedversion{}

\supplement{}

\acknowledgements{}

\nolinenumbers 

\EventEditors{Rolf Niedermeier and Christophe Paul}
\EventNoEds{2}
\EventLongTitle{36th International Symposium on Theoretical Aspects of Computer Science (STACS 2019)}
\EventShortTitle{STACS 2019}
\EventAcronym{STACS}
\EventYear{2019}
\EventDate{March 13--16, 2019}
\EventLocation{Berlin, Germany}
\EventLogo{}
\SeriesVolume{126}
\ArticleNo{6} 

  \newtheorem{property}[theorem]{Property}

\renewcommand{\phi}{\varphi}

\newcommand{\set}[1]{\left\{#1\right\}}
\newcommand{\NN}{{\mathbbm{N}}}
\newcommand{\ZZ}{{\mathbbm{Z}}}
\newcommand{\RR}{{\mathbbm{R}}}
\newcommand{\CC}{{\mathbbm{C}}}
\newcommand{\QQ}{{\mathbbm{Q}}}
\renewcommand{\AA}{{\mathbbm{A}}}

\newcommand{\B}{{\cal B}}

\newcommand{\I}{{\cal I}}
\newcommand{\OO}{{\cal{O}}}

\newcommand{\co}[1]{\overline{#1}}

\newcommand{\re}{{\rm Re}}

\newcommand{\norm}[1]{{\left\lVert#1\right\rVert}}
\newcommand{\heig}[1]{{H(#1)}}
\newcommand{\diag}{\textrm{diag}}

\newcommand{\PSPACE}{\textbf{PSPACE}}
\newcommand{\DOUBLEEXP}{\textbf{2EXP}}

\begin{document}
\maketitle
\begin{abstract}
The \emph{Semialgebraic Orbit Problem} is a fundamental reachability question that arises in the analysis of discrete-time linear dynamical systems such as automata, Markov chains, recurrence sequences, and linear while loops.  An instance of the problem comprises a dimension $d\in\mathbb{N}$, a square matrix $A\in\mathbb{Q}^{d\times d}$, and semialgebraic source and target sets $S,T\subseteq \mathbb{R}^d$. 	
The question is whether there exists $x\in S$ and $n\in\mathbb{N}$ such that $A^nx \in T$. 	
	
\noindent The main result	of this	paper is that the Semialgebraic Orbit 
        Problem is decidable for dimension $d\leq 3$.  Our decision
        procedure relies on separation bounds for algebraic numbers as
        well as a classical result of transcendental number
        theory---Baker's theorem on linear forms in logarithms of
        algebraic numbers.  We moreover argue that our main result
        represents a natural limit to what can be decided (with
        respect to reachability) about the orbit of a single matrix.  On the one hand, semialgebraic sets are arguably the largest general class	of subsets of $\mathbb{R}^d$ for which membership is decidable.  On the other hand, previous work has shown that in dimension $d=4$, giving a decision procedure for the special case of the Orbit Problem with singleton source set $S$ and polytope target set $T$ would entail major breakthroughs in Diophantine approximation.
	
\end{abstract}
\pagebreak
\section{Introduction}
\label{sec: intro}
This paper concerns decision problems of the following form: given
$d\in\mathbb{N}$, a square matrix $A \in \mathbb{Q}^{d\times d}$, and
respective \emph{source} and \emph{target} sets
$S,T\subseteq\mathbb{R}^d$, does there exist $n\in \mathbb{N}$ and
$x\in S$ such that $A^n x \in T$?  One way to categorise such problems
is according to the types of sets allowed for the source and target
(e.g., polytopes or semialgebraic sets).  We collectively refer to
the various problems that arise in this way as \emph{Orbit Problems}.
Orbit Problems occur naturally in the reachability analysis of
discrete-time linear dynamical systems, including Markov chains,
automata, recurrence sequences, and linear loops in program analysis
(see~\cite{ChonevOW16,kannan1986polynomial,TUCS05} and references therein).

In order to describe the main result of this paper in relation to
existing work, we identify three successively more general types of
Orbit Problems.  In the \emph{point-to-point} version both the
source and target are singletons with rational coordinates; in the
\emph{Polytopic Orbit Problem} the source and target $S$ and $T$ are
polytopes (i.e., sets defined by conjunctions of linear inequalities
with rational coefficients); in the \emph{Semialgebraic Orbit Problem} $S$ and $T$ are semialgebraic sets defined with rational
parameters.

The question of the decidability of the point-to-point Orbit Problem
was raised by Harrison in 1969~\cite{harrison1969lectures}.  The
problem remained open for ten years until it was finally resolved in a
seminal paper of Kannan and Lipton~\cite{kannan1986polynomial}, who in
fact gave a polynomial-time decision procedure.

The Polytopic Orbit Problem is considerably more challenging than the
point-to-point version, and its decidablity seems out of reach for
now.  Indeed the special case in which $S$ is a singleton and $T$ is a
linear subspace of $\mathbb{R}^d$ of dimension $d-1$ is a well-known
decision problem in its own right, called the \emph{Skolem Problem},
whose decidability has been open for many
decades~\cite{tao2008structure}.  In contrast to the point-to-point
case the only positive decidability results for the Polytopic Orbit
Problem are in the case of fixed dimension $d$.  For the Skolem Problem,
decidability is known for $d \leq 4$~\cite{MST84,Ver85}.  In case $S$
and $T$ are allowed to be arbitrary polytopes, decidability is known
in case $d\leq 3$~\cite{AOW17} (see also~\cite{chonev2015polyhedron}).
While Kannan and Lipton's decision procedure in the point-to-point
case mainly relied on algebraic number theory (e.g., separation bounds
between algebraic numbers and prime factorisation of ideals in rings
of algebraic integers), the decision procedures for the Skolem Problem
and the Polytopic Orbit Problem additionally use results about
transcendental numbers (specifically Baker's theorem about linear
forms in logarithms of algebraic numbers).  It was shown
in~\cite{chonev2015polyhedron} that the existence of a decision
procedure for the Polytopic Orbit Problem in dimension $d=4$ would
entail computability of the Diophantine approximation types of a
general class of transcendental numbers (a problem considered
intractable at present).  Not only does this suggest that the use of
transcendental number theory is unavoidable in analysing the Polytopic
Orbit Problem, it also indicates that further progress beyond the case
$d=3$ is contingent upon significant advances in the field of Diophantine
approximation.

In this paper we remain in dimension $d=3$ and consider a
generalisation of previous work by allowing the source and target sets
to be semialgebraic, that is, defined by Boolean combinations of
polynomial equalities and inequalities.  This allows us to handle
three-dimensional source and target sets in much greater geometrical
generality than polytopes.  In applications to program analysis and dynamical systems,
semialgebraic sets are indispensable in formulating sufficiently expressive models  (e.g., to describe
initial conditions and transition guards) and in model analysis (e.g., in synthesising invariants 
and barrier certificates and approximating sets of reachable states)~\cite{Muller-OlmS04,LafferrierePY01}.

The Semialgebraic Orbit Problem could be reduced
to the polytopic case in a fairly straightforward fashion by
increasing the dimension $d$ according to the degree of the
polynomials appearing in the semialgebraic constraints.  However such
a general approach is doomed to failure in view of the obstacles to
obtaining decidability in the polytopic case beyond dimension $3$ and
instead we develop specific techniques for the semialgebraic case that
are considerably more challenging than in the Polytopic Problem.  As
in previous work on the Skolem Problem and on the Polytopic Orbit
Problem, Baker's Theorem plays a crucial role in the present
development.  The main difficulty in generalising from the polytopic
case to the semialgebraic case lies in the delicate analytic arguments
that are required to bring Baker's Theorem to bear. More precisely:
(i)~we need to resort to symbolic quantifier elimination (in lieu of
explicit Fourier-Motzkin elimination, which had been used in the
Polytopic Orbit Problem), since we are now dealing with non-linear
constraints; (ii)~we also need to perform spectral calculations
symbolically, via the use of Vandermonde methods, instead of the
explicit direct approach possible in our earlier work; and (iii)~we replace
triangulation of polytopes by cylindrical algebraic decomposition of
semialgebraic sets into cells, which again necessitates a new symbolic
treatment along with a substantially refined analysis based on Taylor
approximation of the attendant functions.

In summary, this paper provides a decision procedure for the
Orbit Problem in dimension $d=3$ with semialgebraic source and target sets.  
The latter appear to be a natural limit to the
positive decidability results that can be obtained for this problem,
barring major new advances in Diophantine approximation.

At a technical level, our contributions are twofold: in
Section~\ref{sec: semialgebraic hitting} we start by analysing the
case of the Orbit Problem in which $S$ is a singleton and $T$ a
semialgebraic set. We then reduce this problem in Section~\ref{subsec:
  hitting to system} to solving certain systems of
polynomial-exponential equalities and inequalities, and in
Section~\ref{subsec: solving the systems} we show precisely how to solve such
systems. The second technical contribution consists in handling the general
case of the Semialgebraic Orbit Problem, in Section~\ref{sec: semialgebraic collision}. There, we show how to circumvent problems
that arise when quantifying over the set $S$, and arrive at a system
that can ultimately be solved using the techniques and results
developed in Section~\ref{subsec: solving
  the systems}.
\section{Mathematical Tools}
In this section we introduce the key technical tools used in this paper.
\subsection{Algebraic numbers}
\label{sec:algebraic numbers}
For $p\in \ZZ[x]$ a polynomial with integer coefficients, we denote by $\norm{p}$ the bit length of its representation as a list of coefficients encoded in binary. 
Note that the {\em degree} of $p$, denoted $\deg(p)$ is at most $\norm{p}$, and the {\em height} of $p$ --- i.e., the maximum of the absolute values of its coefficients, denoted $\heig{p}$ --- is at most $2^\norm{p}$.

We begin by summarising some basic facts about the field of 
algebraic numbers (denoted $\AA$) and (efficient) arithmetic therein.
The main references include~\cite{basu2005algorithms,cohen2013course, renegar1992computational}.
A complex number $\alpha$ is {\em algebraic} if it is a root of a single-variable polynomial with integer coefficients. The
{\em defining polynomial} of $\alpha$, denoted $p_\alpha$, is the unique polynomial of least degree, and whose coefficients do not have common factors, which vanishes at $\alpha$. The {\em degree} and {\em height} of $\alpha$ are respectively those of $p$, and are denoted $\deg(\alpha)$ and $\heig{\alpha}$. A standard representation\footnote{Note that this representation is not unique.} for algebraic numbers is to encode $\alpha$ as a tuple comprising its defining polynomial
together with rational approximations of its real and imaginary parts of sufficient precision to distinguish $\alpha$ from the other roots of $p_\alpha$. More precisely, $\alpha$ can be represented by $(p_\alpha, a, b, r)\in \ZZ[x]\times \QQ^3$ provided that $\alpha$ is the unique root of $p_\alpha$ inside the circle in $\CC$ of radius $r$ centred at $a + bi$. A  separation bound due to Mignotte~\cite{mignotte1983some} asserts that for roots $\alpha\neq \beta$ of a polynomial $p\in \ZZ[x]$,  we have
\vspace*{-2pt}
\begin{equation}
\label{eq:Mignotte}
|\alpha-\beta|>\frac{\sqrt{6}}{d^{(d+1)/2}H^{d-1}}
\end{equation}
where $d=\deg(p)$ and $H=\heig{p}$. Thus if $r$ is required to be less than a quarter of the root-separation bound, the representation is well-defined and allows for equality checking.
Given a polynomial $p\in \ZZ[x]$, it is well-known how to compute standard representations of each of its roots in time  polynomial in $\norm{p}$~\cite{basu2005algorithms,cohen2013course,pan1996optimal}. Thus given an algebraic number $\alpha$ for which we have (or wish to compute) a
standard representation, we write $\norm{\alpha}$ to denote the bit length of this representation. From now on, when referring to computations on algebraic numbers, we always implicitly refer to their standard representations.

Note that Equation~\eqref{eq:Mignotte} can be used more generally to separate arbitrary algebraic numbers: indeed, two algebraic numbers $\alpha$ and $\beta$ are always roots of the polynomial $p_\alpha p_\beta$ of degree at most $\deg(\alpha)+\deg(\beta)$, and of height at most $\heig{\alpha}\heig{\beta}$. 
Given algebraic numbers $\alpha$ and $\beta$, one can compute
$\alpha+\beta$, $\alpha\beta$, $1/\alpha$ (for $\alpha\neq 0$), $\co{\alpha}$, and $|\alpha|$, all of which are algebraic, in time polynomial in $\norm{\alpha}+\norm{\beta}$. Likewise, it is straightforward to check whether $\alpha=\beta$.
Moreover, if $\alpha\in\RR$, deciding whether $\alpha>0$ can be done in time polynomial in $\norm{\alpha}$. Efficient  algorithms for all these tasks can be found in~\cite{basu2005algorithms,cohen2013course}.

\subsection{First-order theory of the reals}
\label{sec:first order theory}
Let $\vec{x}=x_1,\ldots, x_m$ be a list of $m$ real-valued variables, and let $\sigma(\vec{x})$ be a Boolean combination of atomic predicates of the form $g(\vec{x}) \sim 0$, where each $g(\vec{x})\in \ZZ[x]$ is a polynomial with integer coefficients over these variables, and ${\sim}\in \set{>,=}$. A {\em formula of the first-order theory of the reals} is of the form 
$Q_1x_1 Q_2x_2\cdots Q_mx_m \sigma(\vec{x})$,
where each $Q_i$ is one of the quantifiers $\exists$ or $\forall$. Let us denote the above formula by $\tau$ , and write $\norm{\tau}$ to denote the bit length of its syntactic representation.
Tarski famously showed that the first-order theory of the reals is decidable~\cite{tarski1951decision}. His procedure, however, has non-elementary complexity. Many substantial improvements followed over the years, starting with Collins’s technique of cylindrical algebraic decomposition~\cite{collins1975quantifier}, and culminating with the fine-grained analysis of Renegar~\cite{renegar1992computational}. 
In this paper, we will use the following theorems~\cite{Renegar88, renegar1992computational}.
\begin{theorem}[Renegar~\cite{Renegar88}]
	\label{thm:renegar existential}
	The problem of deciding whether a closed formula $\tau$ of the form above holds over the reals is in \DOUBLEEXP, and in \PSPACE\ if $\tau$ has only existential quantifiers. 
\end{theorem}
\begin{theorem}[Renegar~\cite{renegar1992computational}]
	\label{thm:renegar quantifier elimination}
	There is an algorithm that, given a formula $\tau(x_1,\ldots,x_m)$ where $x_1,\ldots,x_m$ are free variables, computes an equivalent quantifier-free formula in \emph{disjunctive normal form} (DNF) $\Phi(x_1,\ldots,x_m)=\bigvee_{I}\bigwedge_{J} R_{I,J}(x_1,\ldots,x_m)\sim_{I,J}0$ where $R_{I,J}$ is a polynomial \footnote{Technically, the indices should be $I,J_I$, but we omit the dependency of $J$ on $I$ for simplicity.} and $\sim_{I,J}\in \set{>,=}$. 
	Moreover, the algorithm runs in time $2^{2^{\OO(\norm{\tau})}}$, and in particular, $\norm{\Phi}=2^{2^{\OO(\norm{\tau})}}$.
\end{theorem}

A set $S\subseteq \RR^d$ is \emph{semialgebraic} if there exists a formula $\Phi(x_1,\ldots,x_d)$ in the first-order theory of the reals with free variables $x_1,\ldots,x_d$ such that $S=\set{(c_1,\ldots,c_d): \Phi(c_1,\ldots,c_d) \mbox{ is true}}$. 

We remark that algebraic constants can also be incorporated as coefficients in the first-order theory of the reals (and in particular, in the definition of semialgebraic sets), as follows. Consider a polynomial $g(x_1,\ldots,x_m)$ with algebraic coefficients $c_1,\ldots, c_k$. We replace every $c_j$ with a new, existentially-quantified variable $y_j$, and add to the sentence the predicates $p_{c_j}(y_j)=0$ and $(y_j-(a+bi))^2< r^2$, where $(p_{c_j},a,b,r)$ is the representation of $c_j$. Then, in any evaluation of this formula to True, it must hold that $y_j$ is assigned value $c_j$.

\section{Almost Self-Conjugate Systems of Inequalities}
\label{sec: systems}
\label{sec: semialgebraic hitting}
In this section we lay the groundwork for solving the Semialgebraic Orbit Problem. We do so by initially treating the case where the set $S$ of initial points is a singleton.
\subsection{Analysis of the Point-to-Semialgebraic Orbit Problem}
\label{subsec: hitting to system}
The \emph{point-to-semialgebraic Orbit Problem} is 
to decide, given a matrix $A\in \QQ^{3\times 3}$, an initial point $s\in \QQ^3$ and a semialgebraic target $T\subseteq \RR^3$, whether there exists $n\in \NN$ such that $A^n s\in T$.

By Theorem~\ref{thm:renegar quantifier elimination}, we can compute a quantifier-free representation of $T$. That is, we can write $T=\{(x,y,z):\bigvee_{I}\bigwedge_{J}R_{I,J}(x,y,z)\sim_{I,J} 0\}$ where $R_{I,J}$ are polynomials with integer coefficients, and $\sim_{I,J}\in \set{>,=}$. 
For the purpose of solving the point-to-semialgebraic Orbit Problem, we note that it is enough to consider each disjunct separately. Thus, we can assume  $T=\set{(x,y,z):\bigwedge_{J}R_J(x,y,z)\sim_J 0}$, and it remains to decide whether there exists $n\in \NN$ such that $\bigwedge_J R_J(A^ns)\sim_J 0$. 

Note that, as per Theorem~\ref{thm:renegar quantifier elimination}, we have that $\norm{R_J}=2^{2^{\OO(\norm{T})}}$. Moreover, the number of terms in the DNF formula above can itself be doubly-exponential in $\norm{T}$. Complexity wise, this is the most expensive part of our algorithm.

Consider the eigenvalues of $A$. Since $A$ is a $3\times 3$ matrix, then either it has only real eigenvalues, or it has one real eigenvalue and two conjugate complex eigenvalues. In particular, if $A$ has complex eigenvalues, then it is diagonalisable. 

The case where $A$ has only real eigenvalues is treated in Appendix~\ref{apx:real eigen} for the general case of the Semialgebraic Orbit Problem, and is considerably simpler. 

Henceforth, we assume $A$ has complex eigenvalues, so that $A=PDP^{-1}$ with 
$D=\begin{pmatrix}
\lambda & 0 &0 \\
0& \co{\lambda} &0\\
0&0 & \rho
\end{pmatrix}$, where $\lambda$ is a complex eigenvalue, $\rho\in \RR$, and $P$ an invertible matrix. 

Observe that $A^n=PD^nP^{-1}$. By carefully analysing the structure of $P$, it is not hard to show that $A^ns=\begin{pmatrix}
a_1 \lambda^n+\co{a_1}\co{\lambda}^n+b_1 \rho^n\\
a_2 \lambda^n+\co{a_2}\co{\lambda}^n+b_2 \rho^n\\
a_3 \lambda^n+\co{a_3}\co{\lambda}^n+b_3 \rho^n
\end{pmatrix}$
where the $a_i$ and $b_i$ are algebraic and the $b_i$ are also real (see Appendix~\ref{apx: similarity matrix} for a detailed analysis). 

Thus, we want to decide whether there exists $n\in \NN$ such that $R_J(a_1 \lambda^n+\co{a_1}\co{\lambda}^n+b_1 \rho^n,
a_2 \lambda^n+\co{a_2}\co{\lambda}^n+b_2 \rho^n,
a_3 \lambda^n+\co{a_3}\co{\lambda}^n+b_3 \rho^n)\sim_J 0$ for every $J$. Since $R_J$ is a polynomial, then by aggregating coefficients we can write
\begin{equation}
\label{eq: poly expression}
R_J(A^ns)=\sum_{0\le p_1,p_2,p_3\le k} \alpha_{p_1,p_2,p_3} \lambda^{n p_1}\co{\lambda}^{n p_2}\rho^{n p_3}+\co{\alpha_{p_1,p_2,p_3}} \co{\lambda}^{n p_1}\lambda^{n p_2}\rho^{n p_3}
\end{equation}
for some $k\in \NN$. Note that we treat the (real) coefficients of $\rho$ as a sum of complex conjugate coefficients, but this can easily be achieved by writing e.g., $c\rho^{n p}=\frac{c}{2}\rho^{n p}+\frac{c}{2}\rho^{n p}$. 

We notice that for every $J$, the polynomial $R_J(A^ns)$, consists of conjugate summands. More precisely, $R_J(A^n s)$, when viewed as a polynomial in $\lambda^n,\co{\lambda}^n,$ and $\rho^n$, has the following property.
\begin{property}[Almost Self-Conjugate Polynomial]
A complex polynomial $Q(z_1,z_2,z_3)$ over $\CC^3$ is \emph{almost self-conjugate} if
\[Q(z_1,z_2,z_3)=\sum_{0\le t_1,t_2,t_3\le \ell} \delta_{t_1,t_2,t_3}z_1^{t_1}z_2^{t_2}z_3^{t_3}+\co{\delta_{t_1,t_2,t_3}}z_2^{t_1}z_1^{t_2}z_3^{t_3}.\]
That is, if $z_2=\co{z_1}$ and $z_3$ is a real variable, then the monomials in $Q$ appear in conjugate pairs with conjugate coefficients.
\end{property}

We refer to the conjunction $\bigwedge_{J}R_J(A^ns)\sim_J 0$ as an \emph{almost self-conjugate system}. It remains to show that we can decide whether there exists $n\in \NN$ that solves the system.

\subsection{Solving Almost Self-Conjugate Systems}
\label{subsec: solving the systems}
Our starting point is now an almost self-conjugate system as described above. In the following, we will consider a single conjunct $R_J(A^n s)\sim_J 0$.

We start by normalising the expression $R_J(A^ns)\sim_J 0$ in the form of~\eqref{eq: poly expression}, as follows. Let $\Lambda=\max\set{|\lambda^{p_1}\co{\lambda}^{p_2}\rho^{p_3}|: \alpha_{p_1,p_2,p_3}\neq \emptyset}$, we divide the expression in~\eqref{eq: poly expression} by $\Lambda^n$, and get that $R_J(A^ns)\sim_J 0$ iff
\begin{equation}
\label{eq: normalised}
\sum_{m=0}^{k} \beta_m \gamma^{n m}+\co{\beta_m}\co{\gamma}^{n m} + r(n)\sim_J 0
\end{equation}
where the $\beta_m$ are algebraic coefficients, $\gamma=\frac{\lambda}{|\lambda|}$ satisfies $|\gamma|=1$ and $r(n)=\sum_{l=1}^{k'} \chi_l \mu^n_{l}+\co{\chi_l}\co{\mu_l}^n$ with $\chi_l$ being algebraic coefficients, and $|\mu_l|<1$ for every $1\le l\le k'$. Moreover, every $\mu_l$ is a quotient of two elements of the form $\lambda^{p_1}\co{\lambda}^{p_2}\rho^{p_3}$, and thus, by Section~\ref{sec:algebraic numbers}, $\deg(\mu_l)=\norm{R_J}^{\OO(1)}$ and $H(\mu_l)=2^{\norm{R_J}^{\OO(1)}}$. 
Note that for simplicity, we reuse the number $k$, although it may differ from $k$ in~\eqref{eq: poly expression}.
We refer to Equation~\eqref{eq: normalised} as the \emph{normalised expression}.

In the following, we assume that at least one of the $\beta_j$ is nonzero for $j\ge 1$. Indeed, otherwise we can recast our analysis on $r(n)$, which is of lower order.

We now split our analysis according to whether or not $\gamma$ is a root of unity. That is, whether $\gamma^d=1$ for some $d\in \NN$.
\subsubsection{The case where $\gamma$ is a root of unity}
\label{sec:root of unity}
Suppose that $\gamma$ is a root of unity. Then, the set $\set{\gamma^n: n\in \NN}$ is a finite set $\set{\gamma^0,\ldots,\gamma^{d-1}}$. Thus, by splitting the analysis of $A^ns$ according to $n\!\!\mod d$, we can reduce the problem to $d$ instances which involve only real numbers. In Appendix~\ref{apx:root of unity} we detail how to handle this case, and comment on its complexity.

\subsubsection{The case where $\gamma$ is not a root of unity}
When $\gamma$ is not a root of unity, the set $\set{\gamma^n:n\in \NN}$ is dense in the unit circle. With this motivation in mind, we define, for a normalised expression, its {\em dominant function} $f:\CC\to\RR$ as 
$f(z)=\sum_{m=0}^{k} \beta_m z^m+\co{\beta_m}\co{z}^{m}$. Observe that \eqref{eq: normalised} is now equivalent to $f(\gamma^n)+r(n)\sim_J 0$.

Our main technical tool in handling~\eqref{eq: normalised} is the following lemma.

\begin{lemma}
\label{lem:main lemma}
	Consider a normalised expression as in \eqref{eq: normalised}. Let $\norm{\I}$ be its encoding length, and let $f$ be its dominant function.
	Then there exists $N\in \NN$ computable in polynomial time in $\norm{\I}$ with $N=2^{\norm{\I}^{\OO(1)}}$
	such that for every $n>N$ it holds that 
	\begin{enumerate}
		\setlength\itemsep{-3pt}
		\item $f(\gamma^n)\neq 0$,
		\item $f(\gamma^n)>0$ iff $f(\gamma^n)+r(n)>0$,
		\item $f(\gamma^n)<0$ iff $f(\gamma^n)+r(n)<0$.
	\end{enumerate}
\end{lemma}
In particular, the lemma implies that if $f(\gamma^n)+r(n)=0$, then $n\le N$. That is, if $\sim_J$ is ``$=$'', then there is a bound on $n$ that solves the system.

\begin{remark}
	In the formulation of Lemma~\ref{lem:main lemma}, we measure the complexity with respect to $\norm{\I}$. However, recall that when the input is $T$, we actually have $\norm{\I}=2^{2^{\OO(\norm{T})}}$. The analysis in Lemma~\ref{lem:main lemma} thus allows us to separate the blowup required for analysing the semialgebraic target from our algorithmic contribution. In particular, when the target has bounded description length, we can obtain better complexity bounds.
\end{remark}

We prove Lemma~\ref{lem:main lemma} in the remainder of this section.

Since $\set{\gamma^n:n\in \NN}$ is dense on the unit circle, our interest in $f$ is also about the unit circle. 
By identifying $\CC$ with $\RR^2$, we can think of $f$ as a function of two real variables. In this view, $f(x,y)$ is a polynomial with algebraic coefficients, and we can therefore compute a description of the algebraic set $Z_f=\set{(x,y):f(x,y)=0 \wedge x^2+y^2=1}$. We start by showing that this set is finite. Define $g:(-\pi,\pi]\to \RR$ by $g(x)=f(e^{ix})$. Explicitly, we have $g(x)=\sum_{m=0}^{k} 2|\beta_m|\cos(mx+\theta_m)$ where $\theta_m=\arg(\beta_m)$. Clearly there is a one-to-one correspondence between $Z_f$ and the roots of $g$.

We present the following proposition, which will be reused later in the proof.
\begin{proposition}
	\label{prop: derivatives}
	For every $x\in (-\pi,\pi]$ there exists $1\le i\le 4 k$ such that $g^{(i)}(x)\neq 0$, where $g^{(i)}$  is the $i$-th derivative of $g$.
\end{proposition}
%
%
\begin{proof}
Assume by way of contradiction that $g'(x)=\ldots=g^{4k}(x)=0$.
For every $1\le i\le 4k$ we have that 
\[g^{(i)}(x)=\begin{cases} 
\sum_{m=1}^k m^i 2|\beta_m| \cos(mx+\theta_m)& i\equiv_4 0\\
\sum_{m=1}^k -m^i 2|\beta_m| \sin(mx+\theta_m)& i\equiv_4 1\\
\sum_{m=1}^k -m^i 2|\beta_m| \cos(mx+\theta_m)& i\equiv_4 2\\
\sum_{m=1}^k m^i 2|\beta_m| \sin(mx+\theta_m)& i\equiv_4 3
\end{cases}\]
(note that the summand that corresponds to $m=0$ is dropped out in the derivative, as it is constant).

Splitting according $i\mod 4$, we rewrite the equations $g^{(i)}(x)=0$ in matrix form as follows.\footnote{By splitting modulo 2, we could actually improve the bound in the proposition from $4k$ to $2k$, but this further complicates the proof.} 
\begin{align*}
&\mbox{for }i\equiv_4 0:\hspace*{1cm}\begin{pmatrix}
1^4 & 2^4 & \cdots & k^4\\
1^8 & 2^8 & \cdots & k^8\\
\vdots &\vdots& \vdots& \vdots\\
1^{4k} & 2^{4k} & \cdots & k^{4k}
\end{pmatrix}\begin{pmatrix}
2|\beta_1|\cos(x+\theta_1)\\
2|\beta_2|\cos(2x+\theta_2)\\
\vdots\\
2|\beta_k|\cos(kx+\theta_k)
\end{pmatrix}=\begin{pmatrix}
0\\
0\\
\vdots\\
0
\end{pmatrix} 
\end{align*}
\begin{align*}
&\mbox{for }i\equiv_4 1:\hspace*{1cm}\begin{pmatrix}
-1^1 & -2^1 & \cdots & -k^1\\
-1^5 & -2^5 & \cdots & -k^5\\
\vdots &\vdots& \vdots& \vdots\\
-1^{4k-3} & -2^{4k-3} & \cdots & -k^{4k-3}
\end{pmatrix}\begin{pmatrix}
2|\beta_1|\sin(x+\theta_1)\\
2|\beta_2|\sin(2x+\theta_2)\\
\vdots\\
2|\beta_k|\sin(kx+\theta_k)
\end{pmatrix}=\begin{pmatrix}
0\\
0\\
\vdots\\
0
\end{pmatrix}\\ 
&\mbox{for }i\equiv_4 2:\hspace*{1cm}\begin{pmatrix}
-1^2 & -2^2 & \cdots & -k^2\\
-1^6 & -2^6 & \cdots & -k^6\\
\vdots &\vdots& \vdots& \vdots\\
-1^{4k-2} & -2^{4k-2} & \cdots & -k^{4k-2}
\end{pmatrix}\begin{pmatrix}
2|\beta_1|\cos(x+\theta_1)\\
2|\beta_2|\cos(2x+\theta_2)\\
\vdots\\
2|\beta_k|\cos(kx+\theta_k)
\end{pmatrix}=\begin{pmatrix}
0\\
0\\
\vdots\\
0
\end{pmatrix}\\
&\mbox{for }i\equiv_4 3:\hspace*{1cm}\begin{pmatrix}
1^3 & 2^3 & \cdots & k^3\\
1^7 & 2^7 & \cdots & k^7\\
\vdots &\vdots& \vdots& \vdots\\
1^{4k-1} & 2^{4k-1} & \cdots & k^{4k-1}
\end{pmatrix}\begin{pmatrix}
2|\beta_1|\sin(x+\theta_1)\\
2|\beta_2|\sin(2x+\theta_2)\\
\vdots\\
2|\beta_k|\sin(kx+\theta_k)
\end{pmatrix}=\begin{pmatrix}
0\\
0\\
\vdots\\
0
\end{pmatrix}
\end{align*}

Observe that the matrices we obtain are minors of Vandermonde Matrices (up to their sign), and as such are non-singular~\cite{gantmacher1959theory}. It follows that 
\[\begin{pmatrix}
2|\beta_1|\sin(x+\theta_1)\\
2|\beta_2|\sin(2x+\theta_2)\\
\vdots\\
2|\beta_k|\sin(kx+\theta_k)
\end{pmatrix}=\begin{pmatrix}
0\\
0\\
\vdots\\
0
\end{pmatrix} \mbox{ and } \begin{pmatrix}
	2|\beta_1|\cos(x+\theta_1)\\
	2|\beta_2|\cos(2x+\theta_2)\\
	\vdots\\
	2|\beta_k|\cos(kx+\theta_k)
\end{pmatrix}=\begin{pmatrix}
	0\\
	0\\
	\vdots\\
	0
\end{pmatrix}\]
Recall that we assume at least one $\beta_j$ is nonzero for some $1\le j\le k$, so we have $\cos(jx+\theta_j)=\sin(jx+\theta_j)=0$, which is clearly a contradiction. We thus conclude the proof.
\end{proof}

By Proposition~\ref{prop: derivatives}, it follows that $g$ is not constant, and therefore $f(x,y)$ is not constant on the curve $x^2+y^2=1$. By Bezout's Theorem, we have that $Z_f$ is finite, and consists of at most $4k$ points. Moreover, $f$ is a semialgebraic function (that is, its graph $\set{(x,y,f(x,y)): x,y\in \RR}$ is semialgebraic set in $\RR^3$). Thus, the points in $Z_f$ have semialgebraic coordinates, and we can compute them. By identifying $\RR^2$ with $\CC$, denote $Z_f=\set{z_1,\ldots,z_{4k}}$. 

\begin{remark}
	\label{rmk: algebraic coefficients in f}
	Since the polynomial $f$ has algebraic coefficients, it is not immediately clear how the degree and height of the points in $Z_f$ relate to $\norm{f}$. However, recall that the algebraic coefficients in $f$ are polynomials in the entries of $A^ns$, which are, in turn, algebraic numbers of degree at most $3$ whose description is polynomial in that of $A$ and $s$.
	
	Thus, we can define $Z_f$ with a formula in the first-order theory of the reals with a fixed number of variables. Using results of Renegar~\cite{renegar1992computational}, we show in Appendix~\ref{apx: bounds on Zf} that the points in $Z_f$ have semialgebraic coordinates with description length polynomial in $\norm{f}$.
\end{remark}

We now employ the following lemma from~\cite{ouaknine2014ultimate}, which is itself a consequence of the Baker-W\"ustholz Theorem~\cite{baker1993logarithmic}.
\begin{lemma}[\cite{ouaknine2014ultimate}]
	\label{lem:Baker on unit circle}
	There exists $D\in \NN$ such that for all algebraic numbers $\zeta,\xi$ of modulus 1, and for every $n\ge 2$, if $\zeta^n\neq \xi$, then $|\zeta^n-\xi|>\frac{1}{n^{(\norm{\zeta}+\norm{\xi})^D}}$.
\end{lemma}

Since $\gamma$ is not a root of unity, it holds that $\gamma^{n_1}\neq \gamma^{n_2}$ for every $n_1\neq n_2\in \NN$. Thus,  there exists a computable $N_1\in\NN$ such that $\gamma^n\notin Z_f$ for every $n>N_1$. 
Moreover, by~\cite[Lemma D.1]{ChonevOW16}, we have that $N_1={\norm{f}}^{\OO(1)}$. By Lemma~\ref{lem:Baker on unit circle}, there exists a constant $D\in \NN$ such that for every $n\ge N_1$ and $1\le j\le 4k$ we have that $|\gamma^n-z_j|>\frac{1}{n^{(\norm{f}^D)}}$ (since $\norm{z_j}+\norm{\gamma}=\OO(\norm{f})$). Intuitively, for $n>N_1$ we have that $\gamma^n$ does not get close to any $z_i$ ``too quickly'' as a function of $n$. In particular, for $n>N_1$ we have $f(\gamma^n)\neq 0$. It thus remains to show that for large enough $n$, $r(n)$ does not affect the sign of $f(\gamma^n)+r(n)$. Intuitively, this is the case because $r(n)$ decreases exponentially, while $|f(\gamma^n)|$ is bounded from below by an inverse polynomial. 

For every $z_j\in Z_f$, let $\phi_j=\arg z_j$, so that $f(z)=0$ iff $g(\phi_j)=0$. We assume w.l.o.g. that $\phi_j\in (-\pi,\pi)$ for every $1\le j\le 4k$. Indeed, if $\phi_j=\pi$ for some $j$, then we can shift the domain of $g$ slightly so that all zeros are in the interior. 

For every $1\le j\le 4k$, let $T_j$ be the Taylor polynomial of $g$ around $\phi_j$ such that the degree $d_j$ of $T_j$ is minimal and $T_j$ is not identically 0.  Thus, we have $T_j(x)=\frac{g^{(d_j)}(\phi_j)}{d_j!}(x-\phi_j)^{d_j}$. 
By Proposition~\ref{prop: derivatives} we have that $d_j\le 4k$ for every $j$. In addition, the description of $T_j$ is computable from that of $\norm{f}$.

By Taylor's inequality, we have that for every $x\in [-\pi,\pi]$ it holds that $|g(x)-T_j(x)|\le \frac{M_j|x-\phi_j|^{d_j+1}}{(d_j+1)!}$ where $M_j=\max_{x\in [-\pi,\pi]}\set{g^{(d_j+1)}(x)}$ (where $g$ is extended naturally to the domain $[-\pi,\pi]$). 
By our description of $g^{(d_j+1)}(x)$, we see that $M_j$ is bounded by $M=4k\max_{1\le i\le k}\set{|\beta_i|}k^{4k+1}$.

Let $\epsilon_1>0$ be such that the following conditions hold for every $1\le j\le 4k$.
\begin{enumerate} 
	\item ${\rm sign}(g'(x))$ does not change in $(\phi_j,\phi_j+\epsilon_1)$ nor in $(\phi_j-\epsilon_1,\phi_j)$.
	\item $|g(x)-T_j(x)|\le \frac{1}{2}|T_j(x)|$ for every $x\in (\phi_j-\epsilon_1,\phi_j+\epsilon_1)$.
	\item ${\rm sign}(g'(x))={\rm sign}(T_j'(x))$ for every $x\in (\phi_j-\epsilon_1,\phi_j+\epsilon_1)$.
\end{enumerate}
Note that we can assume $(\phi_j-\epsilon_1,\phi_j+\epsilon_1)\subseteq (-\pi,\pi)$, since by our assumption $\phi_j\in (-\pi,\pi)$ for all $1\le j\le 4k$.

An $\epsilon_1$ as above exists due to the following properties (see Figure~\ref{fig:functions} for an illustration):
\begin{itemize}
	\item There are only finitely many points where $g'(x)=0$, 
	\item $T_j(x)$ is of degree $d_j$, whereas $|g(x)-T_j(x)|$ is upper-bounded by a polynomial of degree $d_j+1$, and
	\item $T_j'(x)$ is the Taylor polynomial of degree $d_j-1$ of $g'(x)$ around $\phi_j$, so by bounding the distance $|g'(x)-T_j'(x)|$ we can conclude the third requirement.
\end{itemize}

\begin{figure}[ht]
	\centering
	\includegraphics[width=0.8\linewidth]{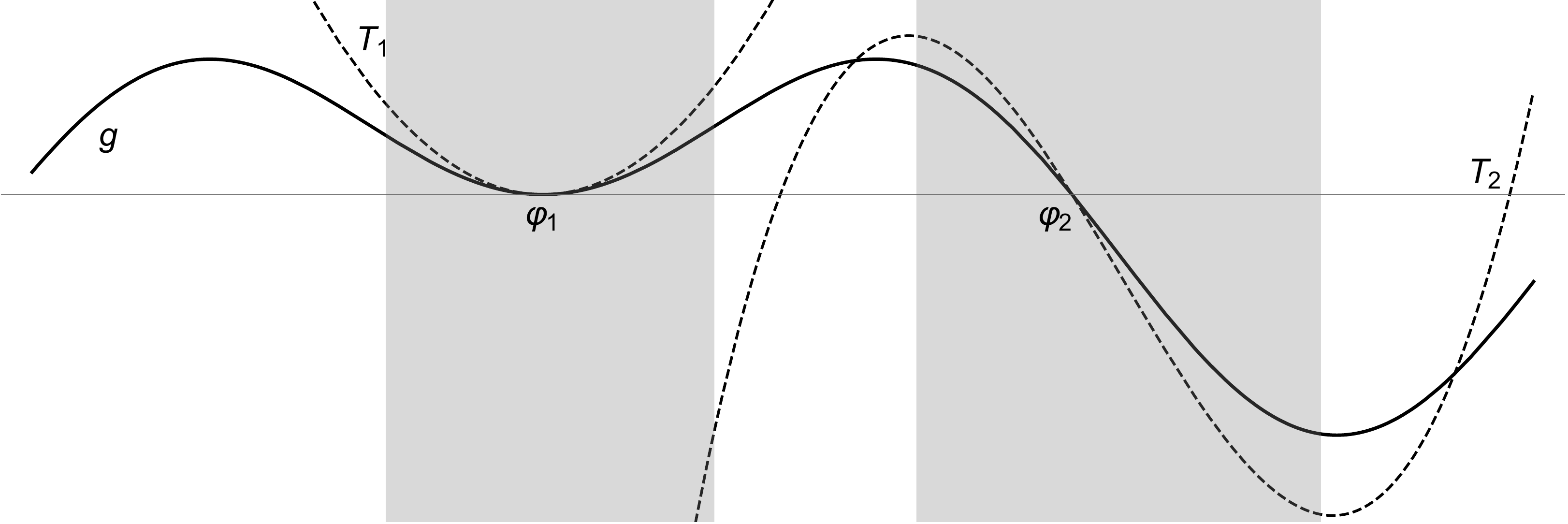}
	\caption{$g(x)$ and two Taylor polynomials: $T_1(x)$ around $\phi_1$ and $T_2(x)$ around $\phi_2$. The shaded regions show where requirements (1)--(3) hold, which determine $\epsilon_1$. Observe that for $T_1$, the most restrictive requirement is $|g(x)-T_1(x)|\le \frac12 T_1(x)$, whereas for $T_2$ the restriction is the requirement that $T_2(x)$ is monotone.}
	\label{fig:functions}
\end{figure}

In order to establish Lemma~\ref{lem:main lemma}, we must be able to effectively compute $\epsilon_1$. We thus proceed with the following lemma.
\begin{lemma}
	\label{lem:epsilon computable}
	$\epsilon_1$ can be computed in polynomial time in $\norm{f}$, and $\frac1{\epsilon_1}=2^{\norm{f}^{\OO(1)}}$.
\end{lemma}
\begin{proof}
	We compute $\delta_1,\delta_2,\delta_3$ that satisfy requirements 1,2, and 3, respectively. Then, taking $\epsilon_1=\min\set{\delta_1,\delta_2,\delta_3}$ will conclude the proof.
	
    \subparagraph*{Condition 1:}
    We compute $\delta_1>0$ such that ${\rm sign}(g'(x))$ does not change in $(\phi_j-\delta_1,\phi_j)$ nor in $(\phi_j,\phi_j+\delta_1)$. This is done as follows. 
    Recall that $g(x)=f(e^{ix})=\sum_{m=0}^k \beta_m e^{imx}+\co{\beta_m}\co{e^{imx}}$. It is not hard to check that $g'(x)=\sum_{m=0}^k im\beta_m e^{imx}+\co{im\beta_m} \co{e^{imx}}$. Let $\widehat{f}(z):\CC\to \RR$ be the function $\widehat{f}(z)=\sum_{m=0}^k im\beta_m z+\co{im\beta_m} \co{z}$, then $g'(x)=\widehat{f}(e^{ix})$ and $\norm{\widehat{f}}=\OO(\norm{f})$.
    
    Consider the algebraic set $F=\set{z: |z|=1\wedge \widehat{f}(z)=0}$, then $\set{x: g'(x)=0}=\set{\arg(z):z\in F}$. By similar arguments as those by which we found the roots of $f$ on the unit circle, namely by adapting Proposition~\ref{prop: derivatives} to $\widehat{f}$, we can conclude that $F$ contains at most $4k$ points. Thus, it is enough to set $\delta_1$ such that $\left(\bigcup_{j=1}^{4k} (\phi_j-\delta_1,\phi_j)\cup (\phi_j,\phi_j+\delta_1)\right)\cap F=\emptyset$.

	By~Equation (\ref{eq:Mignotte}), we have that for $z\neq z'\in F$ it holds that $|z-z'|>\frac{\sqrt{6}}{d^\frac{d+1}{2}\cdot H^{d-1}}$ where $d$ and $H$ are the degree and height of the roots of $\widehat{f}(z)$ (see Remark~\ref{rmk: algebraic coefficients in f}). 
	Thus, $1/|z-z'|$ is $2^{\norm{f}^{\OO(1)}}$, and has a polynomial description. Since $|\arg(z)-\arg(z')|>|z-z'|$, we conclude that by setting $\delta_1=\min\set{|z-z'|: z\neq z'\in F}/3$, it holds that $\frac{1}{\delta_1}$ has a polynomial description in $\norm{f}$, and $\delta_1$ satisfies the required condition.
	
	\subparagraph*{Condition 2:} Next, we compute $\delta_2>0$ such that $|g(x)-T_j(x)|\le \frac{1}{2}|T_j(x)|$ for every $x\in (\phi_j-\delta_2,\phi_j+\delta_2)$. Recall that $T_j(x)=\frac{g^{(d_j)}(\phi_j)}{d_j!}(x-\phi_j)^{d_j}$.
	Note that this case is more challenging than Condition 1, as unlike $g(x)=f(e^{ix})$, the polynomial $T_j(x)$ has potentially transcendental coefficients (namely $\phi_j$).  For clarity, we omit the index $j$ in the following. Thus, we write $T,d,\phi$ instead of $T_j,d_j,\phi_j$, etc.
	
	In order to ignore the absolute value, assume  $T(x)\ge g(x) >0$ in an interval $(\phi,\phi+\xi)$ for some $\xi>0$ (the other cases are treated similarly). Then, the inequality above becomes $g(x)-\frac12 T(x)\ge 0$. 			
	Since the degree of $T$ is $d$, then by the definition of $T$, the first $d-1$ derivatives of $g$ in $\phi$ vanish. Define $h(x)=g(x)-\frac12 T(x)$, then we have $h(\phi)=0$, $h'(\phi)=0,\ldots,h^{(d-1)}(\phi)=0$ and $h^{(d)}(\phi)=g^{(d)}(\phi)-\frac12 g^{(d)}(\phi)=\frac12 g^{(d)}(\phi)$. By our assumption, $T(x)\ge \frac12 T(x)$ for $x\in (\phi,\phi+\xi)$, so $h^{(d)}(\phi)>0$.
	In addition, recall that $|h^{(d+1)}(x)|=|g^{(d+1)}(x)|\le M$ for every $x\in [-\pi,\pi]$. Thus, by writing the $d$-th Taylor expansion of $h(x)$ around $\phi$, we have that $h(x)=\frac{h^{(d)}(\phi)}{d!}(x-\phi)^{d}+E(x)$ where $|E(x)|\le \frac{M}{(d+1)!}(x-\phi)^{d+1}$. We now have that 
	\[h(x)\ge \frac12\frac{g^{(d)}(\phi)}{d!}(x-\phi)^d-\frac{M}{(d+1)!}(x-\phi)^{d+1}.\]
	Taking $x\in (\phi,\phi+\frac{g^{(d)}(\phi)(d+1)}{2 M})$, it is easy to check that $h(x)\ge 0$. We can now set $\delta_2=\frac{g^{(d)}(\phi)(d+1)}{2 M}$, which satisfies the required condition.
	
	\subparagraph*{Condition 3:} Finally, we compute $\delta_3>0$ such that ${\rm sign}(g'(x))={\rm sign}(T'_j(x))$ for every $x\in (\phi_j-\delta_3,\phi_j+\delta_3)$. Observe that $T'_j(x)$ is the $d_j-1$-th Taylor polynomial of $g'(x)$ around $\phi^j$. Thus, by following the reasoning used to find $\delta_2$, we can find $\delta_3$ such that $|g'(x)-T'_j(x)|\le \frac12 |T'_j(x)|$ for every $x\in (\phi-\delta_3,\phi+\delta_3)$, and in particular it holds that  ${\rm sign}(g'(x))={\rm sign}(T'_j(x))$ for every $x\in (\phi_j-\delta_3,\phi_j+\delta_3)$.
	
	As mentioned above, by setting $\epsilon_1=\min\set{\delta_1,\delta_2,\delta_3}$, we conclude the proof.
\end{proof}
Conditions 1,2, and 3 above imply that within the intervals $(\phi_j-\epsilon_1,\phi_j+\epsilon_1)$ we have that $|g(x)|\ge \frac{1}{2}|T_j(x)|$, that $g(x)$ and $T_j(x)$ have the same sign, and that they are both decreasing/increasing together.

We now claim that there exists a polynomial $p(n)$ and a number $N_2\in \NN$ such that for every $n>N_2$ it holds that  $|g(\arg(\gamma^n))|>\frac{1}{p(n)}$. In order to compute $p(n)$, we compute separate polynomials for the domain $\bigcup_{j=1}^{4k} (\phi_j-\epsilon_1,\phi_j+\epsilon_1)$ and for its complement. Then, taking their minimum and bounding it from below by another polynomial yields $p(n)$.

We start by considering the case where $\arg(\gamma^n)\in \bigcup_{j=1}^{4k} (\phi_j-\epsilon_1,\phi_j+\epsilon_1)$. Recall that since $\gamma$ is not a root of unity, then for every $n>N_1$ it holds that $\gamma^n\notin Z_f=\set{z_1,\ldots,z_{4k}}$. Then, by Lemma~\ref{lem:Baker on unit circle}, for every  $1\le j\le 4k$ and every $n\ge N_2=\max\set{N_1,2}$ we have $|\gamma^n-z_j|> \frac{1}{n^{(\norm{f}^D)}}$. In addition, $|\gamma^n-z_j|\le |\arg(\gamma^n)-\phi_j|$ (since the LHS is the Euclidean distance and the RHS is the spherical distance). Therefore, $|\arg(\gamma^n)-\phi_j|>\frac{1}{n^{(\norm{f}^D)}}$, so either $\arg(\gamma^n)>\phi_j+\frac{1}{n^{(\norm{f}^D)}}$ or $\arg(\gamma^n)<\phi_j-\frac{1}{n^{(\norm{f}^D)}}$. Next, we have that if $\arg(\gamma^n)\in (\phi_j-\epsilon_1,\phi_j+\epsilon_1)$ for some $1\le j\le 4k$, then $|g(\arg(\gamma^n))|\ge \frac12 |T_j(\arg(\gamma^n))|\ge \frac12\min\set{{|T_j(\phi_j+\frac{1}{n^{(\norm{f}^D)}})|,|T_j(\phi_j-\frac{1}{n^{(\norm{f}^D)}})|}}$, where the last inequality follows from condition 3 above, which implies that $T_j$ is monotone with the same tendency as $g$. 

Observe that $T_j(\phi_j-\frac{1}{n^{(\norm{f}^D)}})=\frac{g^{(d_j)}(\phi)}{d_j!}\frac{1}{n^{(\norm{f}^D)}}$ and that similarly $T_j(\phi_j+\frac{1}{n^{(\norm{f}^D)}})=-\frac{g^{(d_j)}(\phi)}{d_j!}\frac{1}{n^{(\norm{f}^D)}}$ are both inverse polynomials (in $n$).
Thus, $|g(\arg(\gamma^n))|$ is bounded from below by an inverse polynomial. Moreover, these polynomials can be easily computed in time polynomial in $\norm{f}$.

Finally, we note that for $x\notin \bigcup_{j=1}^{4k} (\phi_j-\epsilon_1,\phi_j+\epsilon_1)$ we can compute in polynomial time a bound $B>0$ such that $|g(x)|>B$. Indeed, $B=\min\{|g(x)|: x\in [-\pi,\pi]\setminus\bigcup_{j=1}^{4k} (\phi_j-\epsilon_1,\phi_j+\epsilon_1)\}$ (where $g(-\pi)$ is defined naturally by extending the domain), and we have that $|B|>0$ since we assumed non of the $\phi_j$ are exactly at $\pi$ (in which case we would have had $g(-\pi)=0$). In particular, we can combine the two domains and compute a polynomial $p$ as required.
We remark that we can compute $\norm{B}$ in polynomial time, since it is either at least $\frac12|T_j(\phi_j\pm \epsilon_1)|$ for some $1\le j\le 4k$ (and by Lemma~\ref{lem:epsilon computable}, $\norm{\epsilon_1}$ can be computed in polynomial time), or it is the value of one of the extrema of $g$, and the latter can be computed by finding the extrema of the (algebraic) function $f$ on the unit circle.

To recap, for every $n>N_2$ it holds that $|g(\arg(\gamma^n))|>\frac{1}{p(n)}$ for a non-negative polynomial $p$, and both $N_2$ and $p$ can be computed in polynomial time in the description of the input.

Next, we wish to find $N_3\in \NN$ such that for every $n>N_3$ it holds that $r(n)<\frac{1}{p(n)}$. Recall that 	$r(n)=\sum_{l=1}^{k'} \chi_l \mu_l^n+ \co{\chi_l} \co{\mu_l}^n$ where for every $1\le l\le k'$ we have that $\mu_l$ is algebraic with $\deg(\mu_l)=\norm{f}^{\OO(1)}$ and $H(\mu_l)=2^{\norm{f}^{\OO(1)}}$. Observe that $1-|\mu_l|$ is also an algebraic number. Indeed, $1-|\mu_l|=1-\sqrt{\mu_l\co{\mu_l}}$. Moreover, we get that $\deg(1-|\mu_l|)\le \deg(\mu_l)^4$, as it is the root of a polynomial of degree at most $\deg(\mu_l)^4$, and that $H(1-|\mu_l|)$ is polynomial in $\heig{\mu_l}$.  
Since $|\mu_l|<1$, by applying Equation~\eqref{eq:Mignotte}, we get $1-|\mu_l|=|1-|\mu_l||>\frac{\sqrt{6}}{d^{(d+1)/2} \heig{\mu_l}^{d-1}}$ where $d=\deg(\mu_l)^{\OO(1)}$ and $\heig{\mu_l}=2^{\norm{I}^{\OO(1)}}$. 
It follows that we can compute $\delta\in (0,1)$ with $\frac{1}{\delta}=2^{\norm{I}^{\OO(1)}}$ such that $1-|\mu_l|>\delta$, and hence $|\mu|^n<1-\delta$. Thus, 
\[|r(n)|\le \sum_{l=1}^{k'} 2|\chi_l||\mu_l|^{mn}\le \sum_{l=1}^{k'} 2|\chi_l|(1-\delta)^{mn}\le 2k'\max_{1\le l\le k'}|\chi_l|(1-\delta)^n\]

We can now compute $\epsilon\in (0,1)$ and $N_3\in \NN$ such that:
\begin{enumerate}
	\item $\frac{1}{\epsilon}=2^{\norm{I}^{\OO(1)}}$
	\item $N_3=2^{\norm{I}^{\OO(1)}}$
	\item For every $n>N_3$ it holds that $|r(n)|<(1-\epsilon)^n$
\end{enumerate}

Finally, by taking $N_4\in \NN$ such that $(1-\epsilon)^n<\frac{1}{p(n)}$ (which satisfies $N_4=2^{\norm{I}^{\OO(1)}}$) for all $n>N_4$, we can now conclude that for every $n>\max\set{N_2,N_3,N_4}$, the following hold.
\begin{enumerate}
	\item $f(\gamma^n)=g(\arg(\gamma^n))\neq 0$.
	\item If $f(\gamma^n)>0$, then $g(\arg(\gamma^n))>0$, so $g(\arg(\gamma^n))>\frac{1}{p(n)}$. Since $|r(n)|<\frac{1}{p(n)}$, it follows that $f(\gamma^n)+r(n)=g(\arg(\gamma^n))+r(n)>\frac{1}{p(n)}-|r(n)|>0$. Conversely, if $f(\gamma^n)+r(n)>0$, then $g(\arg(\gamma^n))+r(n)>0$, but since $|g(\arg(\gamma^n))|>\frac{1}{p(n)}$ and $|r(n)|<\frac{1}{p(n)}$, then it must hold that $g(\arg(\gamma^n))>0$, so $f(\gamma^n)>0$.
	\item If $f(\gamma^n)<0$, then $g(\arg(\gamma^n))<0$, so $g(\arg(\gamma^n))<-\frac{1}{p(n)}$. Since $|r(n)|<\frac{1}{p(n)}$, it follows that $f(\gamma^n)+r(n)=g(\arg(\gamma^n))+r(n)<-\frac{1}{p(n)}+|r(n)|<0$. Conversely, if $f(\gamma^n)+r(n)<0$, then $g(\arg(\gamma^n))+r(n)<0$, but since $|g(\arg(\gamma^n))|>\frac{1}{p(n)}$ and $|r(n)|<\frac{1}{p(n)}$, then it must hold that $g(\arg(\gamma^n))<0$, so $f(\gamma^n)<0$.
\end{enumerate}
This concludes the proof of Lemma~\ref{lem:main lemma}.
\qed

We are now ready to use Lemma~\ref{lem:main lemma} in order to solve the systems.

\begin{theorem}
	\label{thm: solve systmes}
	The problem of deciding whether an almost self-conjugate system has a solution is decidable.
\end{theorem}
\begin{proof} Consider an almost self-conjugate system of the form $\bigwedge_{J}R_J(A^ns)\sim_J 0$.
	For each expression $R_J(A^ns)\sim_J 0$, let $f$ be the corresponding dominant function, as per Lemma~\ref{lem:main lemma}, and compute its respective bound $N$. If $\sim_J$ is ``$=$'', then by Lemma~\ref{lem:main lemma}, if the equation is satisfiable for $n\in \NN$, then $n<N$.
	
	If all the $\sim_J$ are ``$>$'', then for each such inequality compute $\set{z: f(z)>0}$, which is a semialgebraic set. If the intersection of these sets is empty, then if $n$ is a solution for the system, it must hold that $n<N$. If the intersection is non-empty, then it is an open set. Since $\gamma$ is not a root of unity, then $\set{\gamma^n:n\in \NN}$ is dense in the unit circle. Thus, there exists $n>N$ such that $\gamma^n$ is in the above intersection, so the system has a solution. Checking the emptiness of the intersection can be done using Theorem~\ref{thm:renegar existential}.
	
	Thus, it remains to check whether there exists a solution $n<N$, which is clearly decidable.	
\end{proof}

Observe that from Theorem~\ref{thm: solve systmes}, combined with Section~\ref{subsec: hitting to system}, we can conclude the decidability of the point-to-semialgebraic Orbit Problem. However, as it turns out, we can reuse Theorem~\ref{thm: solve systmes} to obtain a much stronger result, namely the decidability of the Semialgebraic Orbit Problem.

\section{The Semialgebraic Orbit Problem}
\label{sec: semialgebraic collision}
In~\cite{AOW17}, we proved that the following problem is decidable: given two polytopes $S,T\subseteq \RR^3$ and a matrix $A\in \QQ^{3\times 3}$, does there exist $n\in \NN$ such that $A^nS\cap T\neq \emptyset$. We now show that the techniques developed here can be used as an alternative solution for this problem, and in fact solve a much stronger variant, where $S$ and $T$ are replaced by semialgebraic sets. That is, given two semialgebraic sets $S,T\subseteq \RR^3$ and a matrix $A\in \QQ^{3\times 3}$, does there exist $n\in \NN$ such that $A^nS\cap T\neq \emptyset$.

\begin{theorem}
	The Semialgebraic Orbit Problem is decidable.
\end{theorem}
\begin{proof}
	Consider semialgebraic sets $S,T\subseteq \RR^3$ and a matrix $A\in \QQ^{3\times 3}$, as described above. Recall that we can write $S=\set{\vec{x}: \bigvee_I \bigwedge_J R_{I,J}(\vec{x})\sim_{I,J}0}$ and similarly for $T$. Since we want to decide whether some point in $S$ hits $T$, we can consider each disjunct in the description of $S$ separately. Thus, we henceforth assume $S=\set{\vec{x}: \bigwedge_J R_{J}(\vec{x})\sim_{J}0}$.
	
	We now turn to characterise the set $A^nS$ for every $n\in \NN$. For this purpose, we assume $A$ is invertible. The case where $A$ is not invertible can be reduced to analysis in a lower dimension, and is handled in Appendix~\ref{apx: singular}.
	For every $n\in \NN$, we now have 
	\[A^nS=\set{A^n \vec{x}: \bigwedge_J R_{J}(\vec{x})\sim_{J}0}=\set{\vec{x}: \bigwedge_J R_{J}((A^{-1})^n\vec{x})\sim_{J}0}.\]
	
	We further assume that $A$ has a complex eigenvalue. As in Section~\ref{sec: semialgebraic hitting}, the case where all eigenvalues are real is simpler (even if $A$ is not diagonalisable), and is handled in Appendix~\ref{apx:real eigen}. We can now write $A=PDP^{-1}$ with 
	$D=\begin{pmatrix}
	\lambda & 0&0\\
	0 & \co{\lambda} &0\\
	0& 0& \rho
	\end{pmatrix}$, where $\lambda$ is a complex eigenvalue, $\rho\in \RR$, and $P$ an invertible matrix. We thus have $A^{-1}=PD^{-1}P^{-1}$ where 
	$D^{-1}=\begin{pmatrix}
	\frac{\co{\lambda}}{|\lambda|^2} & 0 &0\\
	0& \frac{\lambda}{|\lambda|^2} &0\\
	0&0 & \rho^{-1}
	\end{pmatrix}$.
	We denote $\zeta=\frac{\co{\lambda}}{|\lambda|^2}$ and $\eta=\rho^{-1}$, so $D^{-1}=\begin{pmatrix}
	\zeta & 0 &0\\
	0& \co{\zeta} &0\\
	0&0 & \eta
	\end{pmatrix}$. As in Section~\ref{sec: semialgebraic hitting}, by analysing the structure of $P$ and $P^{-1}$, 
	we have that for $\vec{x}=(x_1,x_2,x_3)$,
	$(A^{-1})^n(\vec{x})_i=\sum_{j=1}^3(a_{i,j} \zeta^n+\co{a_{i,j}}\co{\zeta}^n+b_{i,j} \eta^n)x_j$ with $a_{i,j}\in \AA$ and $b_{i,j}\in \AA\cap\RR$. That is, each coordinate $1\le i\le 3$, is a linear combination of $x_1,x_2,x_3$ where the coefficients are of the form above. In particular, the coefficient of every $x_j$ is an almost self-conjugate polynomial (see Appendix~\ref{apx: similarity matrix} for a complete analysis).
	
	Consider a monomial of the form $x_1^{s_1}x_2^{s_2}x_3^{s_3}$ in $R_J(\vec{x})$. Replacing $\vec{x}$ with $(A^{-1})^n\vec{x}$, the monomial then becomes $Q(\zeta^n,\co{\zeta}^n,\eta^n)x_1^{s_1}x_2^{s_2}x_3^{s_3}$, where $Q(z_1,z_2,z_3)$ is an almost self-conjugate polynomial. Indeed, this follows since the coordinates of $(A^{-1})^n\vec{x}$ above are almost self-conjugate, and products of almost self-conjugate polynomials remain almost self-conjugate.
	
	Recall that the polynomials $R_J$ in the description of $S$ have integer (and in particular, real) coefficients. By lifting the discussion about monomials to $R_J$, we can write
	
	\[
	R_J((A^{-1})^n(\vec{x}))=\sum_{0\le s_1,s_2,s_3\le k} Q^J_{s_1,s_2,s_3}(\zeta^n,\co{\zeta}^n,\eta^n)x_1^{s_1}x_2^{s_2}x_3^{s_3}
	\]
	where $k\in \NN$ and the coefficients $Q^J_{s_1,s_2,s_3}$ are almost self-conjugate.
	
	Observe that now, there exists $n\in\NN$ such that $A^nS\cap T\neq \emptyset$ iff there exists $n\in \NN$ and $\vec{x}\in \RR^3$ such that $\vec{x}\in T$ and 
	\begin{equation}
	\label{eq: condition for intersection nonemptiness}
	\bigwedge_{J} \sum_{0\le s_1,s_2,s_3\le k} Q^J_{s_1,s_2,s_3}(\zeta^n,\co{\zeta}^n,\eta^n)x_1^{s_1}x_2^{s_2}x_3^{s_3}\sim_J 0.
	\end{equation}
	Intuitively, we now want to eliminate the quantifiers on $\vec{x}$ in the expression above. However, we cannot readily do so, as the expression is also quantified by $n\in \NN$. Nonetheless, in the following we manage to circumvent this problem by increasing the dimension of the problem.
	
	
%
%
	
	
	Let $K$ be the number of polynomials $Q^J_{s_1,s_2,s_3}$ that appear in the conjunction~\eqref{eq: condition for intersection nonemptiness} above, indexed by $J,s_1,s_2,s_3$. Consider the set
	\[
	U=\set{(y_1,\ldots,y_K)\in \RR^K: \begin{array}{l} \exists \vec{x}\in \RR^3,\ x\in T\wedge \\
		\bigwedge_{J} \sum_{0\le s_1,s_2,s_3\le k} y^J_{s_1,s_2,s_3} x_1^{s_1}x_2^{s_2}x_3^{s_3}\sim_J 0
	\end{array}
	}
	\]
	That is, $U$ is obtained by replacing each polynomial $Q^J_{s_1,s_2,s_3}$ with a ``placeholder'' real variable $y^J_{s_1,s_2,s_3}$.
	$U$ is clearly a semialgebraic set, so by Theorem~\ref{thm:renegar quantifier elimination}, we can eliminate the quantifier on $\vec{x}$, and write
	\[
	U=\set{(y_1,\ldots,y_K)\in \RR^K:\bigwedge_{J}S_J(y_1,\ldots,y_K)\sim_J 0}
	\]
	where $S_J$ are polynomials with integer coefficients.
	It is now the case that there exists $n\in \NN$ such that $A^nS\cap T\neq \emptyset$ iff there exists $n\in \NN$ such that $(Q_1(\zeta^n,\co{\zeta}^n,\eta^n),\ldots,Q_K(\zeta^n,\co{\zeta}^n,\eta^n))\in U$. That is, we need to decide whether there exists $n\in \NN$ such that $S_J(Q_1(\zeta^n,\co{\zeta}^n,\eta^n), \ldots, $ $Q_K(\zeta^n,\co{\zeta}^n,\eta^n))\sim_J 0$ for every $J$.
		
	It is easy to see that since the polynomials $Q_i$ are almost self-conjugate, then so is \(S_J(Q_1(\zeta^n,\co{\zeta}^n,\eta^n),  \ldots ,Q_K(\zeta^n,\co{\zeta}^n,\eta^n)),\) (when viewed as a polynomial in $\zeta^n,\co{\zeta}^n,\eta^n$).
	
	Thus, the conjunction
	\[
	\bigwedge_J S_J(Q_1(\zeta^n,\co{\zeta}^n,\eta^n),  \ldots ,Q_K(\zeta^n,\co{\zeta}^n,\eta^n))
	\]
	is an almost self-conjugate system, and by Theorem~\ref{thm: solve systmes}, it is decidable whether it has a solution. This concludes the proof.
\end{proof}

\section{Discussion}
This paper establishes the decidability of the Semialgebraic Orbit Problem in dimension at most three. The class of semialgebraic sets is arguably the largest natural class for which membership is decidable. Thus, our results reach the limit of what can be decided about the orbit of a single matrix. Moreover, our techniques shed light on the decidability (or hardness) of orbit problems in higher dimensions: the techniques we develop for analysing orbits can be applied to any matrix (in any dimension) whose eigenvalues have arguments that are pairwise linearly dependent over $\QQ$ (i.e., the arguments of all the eigenvalues are rational multiples of some angle $\theta$). Indeed, it is easy to see that the orbits generated by such matrices can be reduced to solving almost self-conjugate systems (see Section~\ref{sec: semialgebraic hitting}).
This can be put in contrast to known hardness results~\cite{chonev2015polyhedron} in dimension $d\ge 4$, which require a single pair of eigenvalues whose arguments do not satisfy the above property. Thus, we significantly sharpen the border of known decidability, and allow future research to focus on hard instances.

Technically, our contribution uncovers two interesting tools. First,
the identification of almost self-conjugate polynomials, and their
amenability to analysis (Section~\ref{sec: semialgebraic hitting}),
and second, the ability to abstract away integral exponents in order to
perform quantifier elimination, by increasing the dimension
(Section~\ref{sec: semialgebraic collision}). The former arises
naturally in the context of matrix exponentiation, while the latter is
an obstacle that is often encountered when quantifying over
semialgebraic sets in the presence of a discrete operator (e.g.,
matrix exponentiation). In the future, we plan to further investigate
the applications of these directions.

\newpage
\bibliography{Main}

\appendix

\section{The case of only real eigenvalues}
\label{apx:real eigen}
In this section we consider the Semialgebraic Orbit Problem in the case where the matrix $A$ has only real eigenvalues, denoted $\rho_1,\rho_2,\rho_3$. In this case, by converting $A$ to Jordan normal form, there exists an invertible matrix $B\in (\AA\cap \RR)^{3\times 3}$ such that one of the following holds:
\begin{enumerate}
	\item $A=B^{-1}\begin{pmatrix}
	\rho_1 & 0 & 0\\
	0 & \rho_2 & 0\\
	0& 0 & \rho_3 \\
	\end{pmatrix}B$, in which case $A^n=B^{-1}\begin{pmatrix}
	\rho_1^n & 0 & 0\\
	0 & \rho_2^n & 0\\
	0& 0 & \rho_3^n \\
	\end{pmatrix}B$.
	\item $A=B^{-1}\begin{pmatrix}
	\rho_1 & 1 & 0\\
	0 & \rho_2 & 0\\
	0& 0 & \rho_3 \\
	\end{pmatrix}B$ with $\rho_1=\rho_2$, in which case $A^n=B^{-1}\begin{pmatrix}
	\rho_1^n & n\rho_1^{n-1} & 0\\
	0 & \rho_1^n & 0\\
	0& 0 & \rho_3^n \\
	\end{pmatrix}B$.
	\item $A=B^{-1}\begin{pmatrix}
	\rho_1 & 1 & 0\\
	0 & \rho_2 & 1\\
	0& 0 & \rho_3 \\
	\end{pmatrix}B$ with $\rho_1=\rho_2=\rho_3$, in which case $A^n=B^{-1}$\\ $\begin{pmatrix}
	\rho_1^n & n\rho_1^{n-1} & \frac12n(n-1)\rho_1^{n-2}\\
	0 & \rho_1^n & n\rho_1^{n-1}\\
	0& 0 & \rho_1^n \\
	\end{pmatrix}B$.
\end{enumerate}

In any of the forms above, we can write 
\[A^n s=\begin{pmatrix}
A_1(n)\rho_1^n+B_1(n)\rho_2^n+C_1(n)\rho_3^n\\
A_2(n)\rho_1^n+B_2(n)\rho_2^n+C_2(n)\rho_3^n\\
A_3(n)\rho_1^n+B_3(n)\rho_2^n+C_3(n)\rho_3^n
\end{pmatrix}\] 
where the $A_i,B_i,$ and $C_i$ are polynomials whose degree is less than the multiplicity of their corresponding eigenvalue.

In Sections~\ref{sec: semialgebraic hitting} and~\ref{sec: semialgebraic collision}, we reduce the problem to finding a solution to an almost self-conjugate system. In the case of real eigenvalues, the notion of almost self-conjugate is meaningless, as there are no complex numbers involved. Thus, following the analysis thereof, and plugging the entries of $A^ns$, we reduce the problem to solving a system of expressions of the form $\bigwedge_J R_J(A^ns)\sim_J 0$, where

\begin{equation}
\label{eq: poly expression real eigen}
R_J(A^n s)= \sum_{0\le p_1,p_2,p_3\le k} \alpha^J_{p_1,p_2,p_3}(n)\rho_1^{p_1 n}\rho_2^{p_2 n}\rho_3^{p_3 n}
\end{equation}
for some $k\in \NN$, and $\alpha^J_{p_1,p_2,p_3}(n)$ are polynomials.

Assuming $\rho_1,\rho_2,\rho_3>0$ (otherwise we can split according to odd and even $n$), for each such expression we can compute a bound $N\in \NN$ based on the rate of growth of the summands, such that either for every $n>N$ the equation holds, or for every $n>N$ it does not hold. 

\section{The case where $\gamma$ is a root of unity}
\label{apx:root of unity}
We assume that $\gamma=\frac{\lambda}{|\lambda|}$ is a root of unity. That is, there exists $d\in \NN$ such that $\gamma^d=1$, so we have that $\set{\gamma^n:n\in \NN}=\set{\gamma^0,\ldots,\gamma^{d-1}}$. 

Let $n\in \NN$ and write $m=(n \!\! \mod d)$. We can now write 
\[A^ns=\begin{pmatrix}
	a_1 |\lambda|^n\gamma^m+\co{a_1}|\lambda|^n\co{\gamma}^m+b_1 \rho^n\\
	a_2 |\lambda|^n\gamma^m+\co{a_2}|\lambda|^n\co{\gamma}^m+b_2 \rho^n\\
	a_3 |\lambda|^n\gamma^m+\co{a_3}|\lambda|^n\co{\gamma}^m+b_3 \rho^n
\end{pmatrix}
=\begin{pmatrix}
2\re(a_1\gamma^m) |\lambda|^n+b_1 \rho^n\\
2\re(a_2\gamma^m) |\lambda|^n+b_2 \rho^n\\
2\re(a_3\gamma^m) |\lambda|^n+b_3 \rho^n
\end{pmatrix}\]
Observe that there exists $n\in\NN$ such that $A^ns\in T$ iff there exists $0\le m\le d-1$ and $r\in \NN\cup\set{0}$ such that $A^{rd+m}s\in T$. We can thus split our analysis according to $m\in \set{0,\ldots,d-1}$. For every such $m$, we need to decide whether there exists $r\in \NN\cup\set{0}$ such that $\begin{pmatrix}
	2\re(a_1\gamma^m) |\lambda|^m(|\lambda|^d)^r+b_1 \rho^m(\rho^d)^r\\
	2\re(a_2\gamma^m) |\lambda|^m(|\lambda|^d)^r+b_2 \rho^m(\rho^d)^r\\
	2\re(a_3\gamma^m) |\lambda|^m(|\lambda|^d)^r+b_3 \rho^m(\rho^d)^r
\end{pmatrix}$
Note that $\gamma^m$, $|\lambda|^m$ and $\rho^m$ are constants. Therefore,
these expressions contain only realalgebraic constants, the system can be viewed as a case handled in the setting of all real eigenvalues. We can thus proceed with the analysis in Section~\ref{apx:real eigen}.

Finally, we remark that $d\le \deg(\gamma)^2$. The proof appears in~\cite{kannan1986polynomial}, and we bring it here for completeness. Since $\gamma$ is a primitive root of unity of order $d$, then the defining polynomial $p_\gamma$ of $\gamma$ is the $d$-th Cyclotomic polynomial, so $\deg(\gamma)=\Phi(d)$, where $\Phi$ is Euler's totient function. Since $\Phi(d)\ge \sqrt{d}$, we get that $d\le \deg(\gamma)^2$. Therefore, the number of cases we consider is polynomial in the original input, and does not involve a blowup in the complexity.

\section{Matrix Forms in Proposition~\ref{prop: derivatives}}
\label{apx: matrix forms}
Recall that we have
\[g^{(i)}(x)=\begin{cases} 
\sum_{m=1}^k m^i 2|\beta_m| \cos(mx+\theta_m)& i\equiv_4 0\\
\sum_{m=1}^k -m^i 2|\beta_m| \sin(mx+\theta_m)& i\equiv_4 1\\
\sum_{m=1}^k -m^i 2|\beta_m| \cos(mx+\theta_m)& i\equiv_4 2\\
\sum_{m=1}^k m^i 2|\beta_m| \sin(mx+\theta_m)& i\equiv_4 3
\end{cases}\]

Writing this in matrix form, split by $i\mod 4$, we have the following.

\begin{align*}
&\mbox{for }i\equiv_4 0:\hspace*{1cm}\begin{pmatrix}
1^4 & 2^4 & \cdots & k^4\\
1^8 & 2^8 & \cdots & k^8\\
\vdots &\vdots& \vdots& \vdots\\
1^{4k} & 2^{4k} & \cdots & k^{4k}
\end{pmatrix}\begin{pmatrix}
2|\beta_1|\cos(x+\theta_1)\\
2|\beta_2|\cos(2x+\theta_2)\\
\vdots\\
2|\beta_k|\cos(kx+\theta_k)
\end{pmatrix}=\begin{pmatrix}
0\\
0\\
\vdots\\
0
\end{pmatrix}\\ 
&\mbox{for }i\equiv_4 1:\hspace*{1cm}\begin{pmatrix}
-1^1 & -2^1 & \cdots & -k^1\\
-1^5 & -2^5 & \cdots & -k^5\\
\vdots &\vdots& \vdots& \vdots\\
-1^{4k-3} & -2^{4k-3} & \cdots & -k^{4k-3}
\end{pmatrix}\begin{pmatrix}
2|\beta_1|\sin(x+\theta_1)\\
2|\beta_2|\sin(2x+\theta_2)\\
\vdots\\
2|\beta_k|\sin(kx+\theta_k)
\end{pmatrix}=\begin{pmatrix}
0\\
0\\
\vdots\\
0
\end{pmatrix}\\ 
&\mbox{for }i\equiv_4 2:\hspace*{1cm}\begin{pmatrix}
-1^2 & -2^2 & \cdots & -k^2\\
-1^6 & -2^6 & \cdots & -k^6\\
\vdots &\vdots& \vdots& \vdots\\
-1^{4k-2} & -2^{4k-2} & \cdots & -k^{4k-2}
\end{pmatrix}\begin{pmatrix}
2|\beta_1|\cos(x+\theta_1)\\
2|\beta_2|\cos(2x+\theta_2)\\
\vdots\\
2|\beta_k|\cos(kx+\theta_k)
\end{pmatrix}=\begin{pmatrix}
0\\
0\\
\vdots\\
0
\end{pmatrix}\\
&\mbox{for }i\equiv_4 3:\hspace*{1cm}\begin{pmatrix}
1^3 & 2^3 & \cdots & k^3\\
1^7 & 2^7 & \cdots & k^7\\
\vdots &\vdots& \vdots& \vdots\\
1^{4k-1} & 2^{4k-1} & \cdots & k^{4k-1}
\end{pmatrix}\begin{pmatrix}
2|\beta_1|\sin(x+\theta_1)\\
2|\beta_2|\sin(2x+\theta_2)\\
\vdots\\
2|\beta_k|\sin(kx+\theta_k)
\end{pmatrix}=\begin{pmatrix}
0\\
0\\
\vdots\\
0
\end{pmatrix}
\end{align*}

\section{The Case where $A$ is Singular}
\newcommand{\zeroB}{\begin{pmatrix}
		0 & 0 \\
		0 & B \\
\end{pmatrix}}
\label{apx: singular}
In this section, we reduce the Semialgebraic Orbit Problem in the case where $A$ is a singular matrix to the case where $A$ is non-singular. Intuitively, we simply cast our analysis to a lower dimension by projecting $A$ on its nonzero eigenvalues. 

In this case, we are given semialgebraic sets $S,T\subseteq\RR^3$ and a matrix $A\in \QQ^{3\times 3}$, where $0$ is an eigenvalue of $A$.

We start with the case where the multiplicity of the eigenvalue $0$ is $1$. Then, we can write $A=P \zeroB P^{-1}$ where $P$ and $B$ are invertible matrices with rational entries.
Indeed, since $0\in \QQ$, then we can decompose $\QQ^3$ as a direct sum $\QQ^3=V_0\oplus V^\bot_0$ where $V_0$ has dimension 1 and $V_0^\bot$ has dimension 2. Let $u\in \QQ^3$ be an eigenvector corresponding to $0$, so that $\textrm{span}(u)=V_0$, and let $v,w\in \QQ^3$ such that $\textrm{span}(v,w)=V_0^\bot$. We now have that $Av,Aw\in V_0^\bot$, so we can write $Av=c_1v+c_2w$ and $Aw=d_1v+d_2 w$ for some $c_1,c_2,d_1,d_2\in \QQ$. Let $P=(u,v,w)$ and $B=\begin{pmatrix}
c_1 & d_1\\
c_2 & d_2
\end{pmatrix}$, then it is easy to verify that $P$ and $B$ are invertible, and that $AP=P\zeroB$, so $A=P\zeroB P^{-1}$, as we wanted.

We now observe the following:
\begin{align*}
&\exists n\in \NN\  \exists x\in S: A^nx\in T & \iff \\
&\exists n\in \NN\ \exists x\in S: P\zeroB P^{-1}x\in T & \iff \\
&\exists n\in \NN\ \exists x'\in P^{-1}S: P\zeroB x'\in P^{-1}T &  \\
\end{align*}
Denote $S'=P^{-1}S$ and $T'=P^{-1}T$, we proceed\footnote{In the following we ignore the case where $n=0$, as this can be checked initially by deciding whether $S\cap T\neq \emptyset$.}:
\begin{align}
\nonumber &\exists n\in \NN\ \exists x'\in S': \zeroB x'\in T' & \iff \\ 
\nonumber &\exists n\in \NN\ \exists x'\in S': \begin{pmatrix} 
0 & 0 \\
0 & B^{n-1} \\
\end{pmatrix}\zeroB x'\in T' & \iff \\ 
&\exists n\in \NN\ \exists x''\in \zeroB S': \begin{pmatrix}
0 & 0 \\
0 & B^{n-1} \\
\end{pmatrix}x''\in T' & \iff 
\label{eq: last condition collision}
\end{align}
Denote $S''=\zeroB S'$, and observe that 
\[
S''=\set{\zeroB \begin{pmatrix}
	y_1\\y_2\\y_3
	\end{pmatrix}: \begin{pmatrix}
	y_1\\y_2\\y_3
	\end{pmatrix}\in S'}=
	\set{\begin{pmatrix}
		0\\B\begin{pmatrix}
		y_2\\y_3
		\end{pmatrix}\end{pmatrix}: 
	\begin{pmatrix}
		0\\y_2\\y_3
	\end{pmatrix}\in S'}=
	\set{\begin{pmatrix}
	0\\z_2\\z_3
	\end{pmatrix}:
	\begin{pmatrix}
	0\\B^{-1}\begin{pmatrix}
	z_2\\z_3
	\end{pmatrix}\end{pmatrix}
	\in S'}
\]
Thus, the vectors in $S''$ have $0$ in their first coordinate. For such vectors, we have the following:
\[
\begin{pmatrix}
0 & 0 \\
0 & B^{n-1} \\
\end{pmatrix}\begin{pmatrix}
0\\z_2\\z_3
\end{pmatrix}\in T' \iff  
B^{n-1}\begin{pmatrix}
z_2\\z_3
\end{pmatrix}\in \set{\begin{pmatrix}
	x_2\\ x_3
	\end{pmatrix}: \begin{pmatrix}
	0\\x_2\\x_3
	\end{pmatrix}\in T'}
\]

Let $S''_2=\set{\begin{pmatrix}
	z_2\\z_3
	\end{pmatrix}:\begin{pmatrix}
	0\\z_2\\z_3
	\end{pmatrix}\in S''}$ and $T'_2=\set{\begin{pmatrix}
	z_2\\z_3
	\end{pmatrix}:\begin{pmatrix}
	0\\z_2\\z_3
	\end{pmatrix}\in T'}$, then the condition in~\eqref{eq: last condition collision} holds iff 
\[
\exists n\in \NN\ \exists \begin{pmatrix}
z_2\\z_3
\end{pmatrix}\in S''_2: B^{n-1} \begin{pmatrix}
z_2\\z_3
\end{pmatrix}\in T'_2
\]

Since $S''_2$ and $T'_2$ are semialgebraic (and are in fact easily computable from $S$ and $T$), we conclude that we can reduce the dimension of the problem.

Next, if the multiplicity of $0$ is $2$, then we can write  $A=P\begin{pmatrix}
0 & 1 & 0\\
0 & 0 & 0\\
0 & 0 & \rho
\end{pmatrix} P^{-1}$ where $\rho$ is a real eigenvalue. Then $A^n=P\begin{pmatrix}
0 & 0 & 0\\
0 & 0 & 0\\
0 & 0 & \rho^n
\end{pmatrix} P^{-1}$ for every $n\ge 2$, and the same approach as above can be taken.

Finally, if the multiplicity of $0$ is $3$, then $A^3=0$, so the problem becomes trivial.

\section{Change of Basis Matrices in the $3\times 3$ case}
\label{apx: similarity matrix}
In this section we consider a diagonalisable matrix $A\in \QQ^{3\times 3}$ with complex eigenvalues. Thus, we can write $A=PDP^{-1}$ with $D=\diag(\lambda,\co{\lambda},\rho)$ with $\lambda\in \AA$ and $\rho\in \AA\cap\RR$.

Note that the columns of the matrix $P$ are eigenvectors of $A$, and moreover, conjugate eigenvalues have conjugate eigenvectors and real eigenvalues have real eigenvectors. We can therefore assume
\[
P=\begin{pmatrix}
a & \co{a} & d\\
b & \co{b} & e\\
c & \co{c} & f
\end{pmatrix}
\]
for $a,b,c\in \AA$ and $d,e,f\in\RR\cap \AA$.

\begin{lemma}
	\label{lem: structure change of basis}
	Let $E=\diag(\delta_1,\delta_2,\delta_3)$ be a diagonal matrix, then every coordinate of $PEP^{-1}$ is of the form
	$\alpha \delta_1+\co{\alpha}\delta_2+\beta\delta_3$, where $\alpha\in \AA$ and $\beta\in \AA\cap \RR$.
\end{lemma}
\begin{proof}

The proof is straightforward: we compute the matrix $P^{-1}$, and then the product $PEP^{-1}$.

We leave it to the reader to verify the following: first, the determinant of $P$ is pure-imaginary, i.e., $\det(P)=mi$ for $m\in \RR\cap \AA$. Second, we have
\[
P^{-1}=\frac{1}{mi}\begin{pmatrix}
f \co{b}-e \co{c} & d \co{c}-f \co{a} & e \co{a}-d \co{b} \\
c e-b f & a f-c d & b d-a e \\
b \co{c}-c \co{b} & c \co{a}-a \co{c} & a \co{b}-b \co{a} \\
\end{pmatrix}
\]

Finally, it is very easy (yet tedious) to verify that $PEP^{-1}$ satisfies the claim. We demonstrate by computing the  coordinate $(PEP^{-1})_{1,2}$.

We have that the first row of $PE$ is $(a\delta_1 , \co{a}\delta_2 , d \delta_3 )$, and hence 
\begin{align*}
(PEP^{-1})_{1,2}&=(PE)_{1,1}P^{-1}_{1,2}+(PE)_{1,2}P^{-1}_{2,2}+(PE)_{1,3}P^{-1}_{3,2}\\
&=\frac{1}{mi}\left(a\delta_1( d \co{c}-f \co{a})+ \co{a}\delta_2( a f-c d)+d \delta_3 (c \co{a}-a \co{c})\right)\\
&=\frac{1}{m}\left(-i \delta_1( ad \co{c}-a f\co{a})+i \delta_2(\co{a}c d - a f \co{a})- i \delta_3 (d  c \co{a}-d a \co{c})\right)
\end{align*}
It is now easy to see that the coefficients of $\delta_1$ and $\delta_2$ are conjugates, and the coefficient of $\delta_3$ is real, as desired.
\end{proof}

\section{Bounds on the Description Size of Points in $Z_f$}
\label{apx: bounds on Zf}
We complete the analysis of Remark~\ref{rmk: algebraic coefficients in f}.

Recall that $f(z)=\sum_{m=0}^{k}\beta_m z^m +\co{\beta_m}\co{z}^m$, and $Z_f=\set{z: f(z)=0 \wedge |z|=1}$. Further recall that for every $0\le m\le k$,  $\beta_m$ is a polynomial in $a_1,a_2,a_3,\co{a_1},\co{a_2},\co{a_3},b_1,b_2,b_3$, where all the latter are linear combinations of roots of the characteristic polynomial of $A$, and are therefore algebraic numbers of degree at most $3$ and description polynomial in $\norm{A}+\norm{s}$.

We can now express the condition $f(z)=0$ using a quantified formula in the first order theory of the reals by replacing each of the constants above (i.e. $a_1$, etc.) by their corresponding description, as per Section~\ref{sec:first order theory}. It follows that in this description, there are at most $9$ variables. We now employ the following result due to Renegar~\cite{renegar1992computational}.

\begin{theorem}[Renegar]
	\label{thm: renegar}
	Let $M \in \NN$ be fixed.
	Let $\tau(\mathbf{y})$ be a formula of the first order theory of the reals.
	Assume that the number of (free and bound) variables in $\tau(\mathbf{y})$ is bounded by $M$.
	Denote the degree of $\tau(\mathbf{y})$ by $d$ and the number of atomic predicates in $\tau(\mathbf{y})$ by $n$.
	
	There is a polynomial time (polynomial in $\norm{\tau(\mathbf{y})}$) procedure which computes an equivalent quantifier-free formula
	\begin{eqnarray*}
		\chi(\mathbf{y}) = \bigvee_{i=1}^I \bigwedge_{j=1}^{J_i} h_{i,j}(y) \sim_{i,j} 0
	\end{eqnarray*}
	where each $\sim_{i,j}$ is either $>$ or $=$, with the following properties:
	\begin{enumerate}
		\item Each of $I$ and $J_i$ (for $1\le i\le I$) is bounded by $(n+d)^{O(1)}$.
		\item The degree of $\chi(\mathbf{y})$ is bounded by $(n+d)^{O(1)}$.
		\item The height of $\chi(\mathbf{y})$ is bounded by $2^{\norm{\tau(\mathbf{y})}(n+d)^{O(1)}}$.
	\end{enumerate}
\end{theorem}

We apply this theorem to the description of $Z_f$ given above, where we identify $\CC$ with $\RR^2$ so that $f$ is indeed a polynomial. Then, we obtain in polynomial time a description of $Z_f$. Moreover, the degrees of the entries is bounded by $\norm{f}^{\OO(1)}$ and their height is bounded by $2^{\norm{f}^{{\OO(1)}}}$.




\end{document}